\documentclass[fleqn]{article}

\usepackage{amsmath}
\usepackage{amssymb}
\usepackage{amsthm}

\usepackage{authblk}

\usepackage{enumerate}
\usepackage{hyperref}

\newtheorem{lmm}{Lemma}
\newtheorem{xmpl}{Example}
\newtheorem{thrm}{Theorem}
\newtheorem{rmrk}{Remark}
\newtheorem{crllr}{Corollary}

\usepackage{amsmath}
\usepackage{amssymb}
\usepackage{amsthm}

\usepackage{tikz}
\usetikzlibrary{decorations.pathreplacing}
\usetikzlibrary{decorations.pathmorphing}
\usetikzlibrary{shapes.geometric}
\usetikzlibrary{arrows}

\tikzset{p0/.style = {shape = circle,    draw, thick, minimum size = 0.6cm}}
\tikzset{p1/.style = {shape = rectangle, draw, thick, minimum size = 0.6cm}}

\tikzset{>=stealth}

\tikzset{every edge/.style = {thick, ->, draw}}
\tikzset{every loop/.style = {thick, ->, draw}}

\renewcommand{\epsilon}{\varepsilon}
\newcommand{\pow}[1]{2^{#1}}
\newcommand{\nats}{\mathbb{N}}
\newcommand{\natspref}[1]{[#1]}
\newcommand{\set}[1]{\{#1\}}

\newcommand{\mem}{\mathfrak{M}}
\newcommand{\game}{\mathcal{G}}
\newcommand{\win}{\mathrm{Win}}
\newcommand{\arena}{\mathcal{A}}

\newcommand{\behavior}{\mathrm{Beh}}

\DeclareMathOperator{\last}{\mathrm{Last}}
\DeclareMathOperator{\update}{\mathrm{Upd}}
\DeclareMathOperator{\nxt}{\mathrm{Nxt}}
\DeclareMathOperator{\init}{\mathrm{Init}}
\newcommand{\pnlty}{p}
\DeclareMathOperator{\waittime}{\mathrm{wt}}
\DeclareMathOperator{\waitvec}{\mathrm{\overline{\waittime}}}
\DeclareMathOperator{\val}{\mathrm{val}}
\newcommand{\weight}{w}
\newcommand{\mpgval}{\nu}
\DeclareMathOperator{\extend}{\mathrm{ext}}
\DeclareMathOperator{\project}{\mathrm{proj}}
\newcommand{\bnd}{b}

\newcommand{\sigmaopt}{\sigma_{\mathrm{opt}}}
\newcommand{\tauopt}{\tau_{\mathrm{opt}}}
\newcommand{\sigmaoptprime}{\sigma_{\mathrm{opt}}'}
\newcommand{\tauoptprime}{\tau_{\mathrm{opt}}'}
\newcommand{\sigmaopthat}{\widehat{\sigma}_{\mathrm{opt}}}

\newcommand{\sigmaopthatprime}{\widehat{\sigma}_{\mathrm{opt}}'}

\newcommand{\valgame}{\val_{\game}}

\newcommand{\exptime}{\textsc{Exptime}}
\newcommand{\np}{\textsc{NP}}
\newcommand{\twoexptime}{\textsc{2Exptime}}
\newcommand{\conp}{\textsc{co-NP}}

\newcommand{\rrcond}{(Q_j,P_j)_{j \in \natspref{k}}}
\newcommand{\tmaxj}{t_{\max_j}}

\allowdisplaybreaks

\begin{document}

\title{Optimal Strategy Synthesis\\for Request-Response Games \thanks{Research partially supported by ANR AVERISS, by the DFG Research Training Group 1298 ``AlgoSyn'', by the CASSTING Project funded by the European Comission's 7th Framework Programme, and by the DFG Transregional Collaborative Research Center 14 ``AVACS''.}}

\author[$\dag$]{Florian Horn}
\author[$\dag\dag$]{Wolfgang Thomas}
\author[$\dag\dag$]{Nico Wallmeier}
\author[$\dag\dag\dag$]{\newline Martin Zimmermann}

\affil[$\dag$]{\small LIAFA,  Universit\'{e} Denis Diderot - Paris 7, 75205 Paris CEDEX 13, France \texttt{florian.horn@liafa.jussieu.fr}}

\affil[$\dag\dag$]{Lehrstuhl f\"{u}r Informatik 7, RWTH Aachen University, 52056 Aachen, Germany \newline \texttt{$\{$thomas, wallmeier$\}$@automata.rwth-aachen.de}}

\affil[$\dag\dag\dag$]{Reactive Systems Group, Saarland University, 66123 Saarbr\"{u}cken, Germany \texttt{zimmermann@react.uni-saarland.de}}

\renewcommand\Authands{ and }

\date{}

\maketitle

\begin{abstract}
\noindent We show the existence and effective computability of optimal winning strategies for request-response games in case the quality of a play is measured by the limit superior of the mean accumulated waiting times between requests and their responses.
\end{abstract}

\section{Introduction}
\label{sec_intro}

Request-response (RR) conditions are ubiquitous in the formal verification of reactive systems, e.g., every request to access a shared resource is eventually granted. Formally, such a condition is expressed as a pair~$(Q, P)$ of state properties, the first one representing the requests and the second one representing the responses. The corresponding request-response condition is satisfied if each time a state in $Q$ is visited, then at this or a later time a state in $P$ is visited, i.e., every request is answered by a response. In linear temporal logic, this requirement is formalized as $\mathbf{G}(Q \rightarrow \mathbf{F}P)$. In practice, one considers conjunctions of request-response conditions which occur in conjunction with safety conditions. In the following, we assume (w.l.o.g.) the state space to be restricted to those states satisfying the safety conditions. Alternatively, one can encode a safety condition as a request-response condition with empty set of responses. 

Synthesis for RR conditions was investigated in a series of works~\cite{ChatterjeeHenzingerHorn11, WallmeierHuttenThomas03, Zimmermann09} considering request-response games. The winning condition of such a game is a conjunction of request-response conditions, i.e., it is given by a finite family~$(Q_j, P_j)_{j \in [k]}$ of $k$ request-response pairs. Player~$0$ wins a play (an infinite path through the finite game graph) if the request-reponse condition~$(Q_j, P_j)$ is satisfied for every $j$. Wallmeier et al.~\cite{WallmeierHuttenThomas03} presented a reduction from RR to B\"{u}chi games using an exponentially-sized memory structure and thereby gave an $\exptime$-algorithm and an exponential upper bound on the memory requirements for both players. Furthermore, they proved an exponential lower bound on the memory requirements for Player~$0$. These results where complemented by Chatterjee et al.~\cite{ChatterjeeHenzingerHorn11} who proved $\exptime$-completeness of solving RR games and tight exponential lower bounds on the memory requirements for both players.

In request-response games, there is an intuitive notion of the waiting time between a request and its earliest response, which can be used to measure the quality of plays and strategies (from Player~$0$'s point of view). There are several ways to aggregate the waiting times of a play to measure the quality of this play. A simple one is to take the maximal waiting time reached during the play and measure the quality of a strategy in terms of the maximal waiting time it allows during a play that is consistent with the strategy. It is straightforward to show that a finite-state winning strategy of size~$s$ bounds the waiting times by $n s$, where $n$ is the number of vertices of the game graph. Thus, the exponential upper bound on memory requirements in RR games also yields an exponential upper bound on the maximal waiting time during plays consistent with this strategy. On the other hand, there are games witnessing matching exponential lower bounds on the maximal waiting time allowed by winning strategies. These results are presented in Section~\ref{sec_rrgames}. 

However, our main result pertains to a more sophisticated way of aggregating waiting times along a play: the value of a play is defined to be the limit superior of the mean accumulated waiting times of the play. In turn, the value of a strategy is the supremum over the values of all plays that are consistent with it. These considerations add a quantitative aspect to the synthesis problem that goes beyond the mere satisfaction of the winning condition and deciding whether there is a bound on the waiting time, by asking for an optimal winning strategy. Our main result states that an optimal strategy always exists, can be presented as a finite-state strategy, and can be effectively computed. To show this, we first prove an upper bound on the value of optimal strategies. Then, we show that for every strategy whose value is smaller than the bound (in particular optimal ones) there is a strategy of smaller or equal value that bounds the waiting times by some doubly-exponential bound. Thus, the search space for an optimal winning strategy is finite and the problem of finding one can be reduced to computing an optimal strategy for a mean-payoff game of doubly-exponential size which encodes the search space. 

Our result fits into a larger series of works which aim at extending synthesis from a decision problem to an optimization problem by asking for optimal winning strategies according to a given quality measure for the winning condition under consideration, e.g., the use of mean-payoff
objectives and weighted automata to model quantitative aspects in the winning
condition~\cite{BCHJ09, CernyChatterjeeHenzingerRadhakrishnaSingh11,
ChatterjeeHenzingerJurdzinski05} and quantitative strengthenings of parity and Streett conditions~\cite{ChatterjeeHenzingerHorn09, FijalkowZimmermann12}. In another line of research, linear temporal logic is extended by parameterized operators~\cite{AlurEtessamiLaTorrePeled01, KupfermanPitermanVardi09} equipped with variables that bound their scope, e.g., the formula~$\mathbf{G}(Q \rightarrow \mathbf{F}_{\le x}P)$ of parametric linear temporal logic (PLTL) is satisfied, if there is some value~$\alpha(x)$ such that every request is answered within $\alpha(x)$ steps. Thus, measuring the quality of plays and strategies in RR games using the maximal waiting time is expressible in PLTL. Asking whether there exists a variable valuation and a winning strategy for Player~$0$ in a game with a PLTL winning condition is $\twoexptime$-complete, while optimal winning strategies can be computed in triply-exponential time~\cite{Zimmermann13}.

Finally, there has been a lot of interest in so-called energy games, whose
winning conditions ask for the existence of an initial amount of energy such that a positive energy level is maintained throughout the play, where energy is consumed or recharged while traversing edges of the game graph. Solving energy games with multiple resources is in general
intractable~\cite{FahrenbergJuhlLarsenSrba11} while so-called consumption
games, a subclass of energy games, are shown to be tractable in~\cite{BCKN12}.
Energy parity games, whose winning conditions are a conjunction
of a (single resource) energy and a parity condition, can be solved in $\np
\cap \conp$ and one player (the spoiling one) has positional winning
strategies while the other needs exponential
memory~\cite{ChatterjeeDoyen10}. 

The paper is structured as follows: in Section~\ref{sec_defs}, we introduce basic definitions about infinite games. In Section~\ref{sec_rrgames}, we introduce RR games, define waiting times and the induced quality measure and prove some preliminary results. In Section~\ref{sec_bounding}, we show that for every strategy of small value there is a strategy of smaller or equal value that additionally bounds the waiting times by some doubly-exponential bound. To this end, we give a quantitive version of Dickson's Lemma in Subsection~\ref{subsec_dickson} and use this to obtain an upper bound in Subsection~\ref{subsec_bounding}. Using this upper bound, we are able to construct a mean-payoff game whose optimal strategy induces an optimal strategy for the RR game. This reduction is presented in Section~\ref{sec_computingoptimalstrategies}. We conclude in Section~\ref{sec_conc} with a discussion and some open questions.

The present paper is a revised version with simplified proofs of results announced in the conference paper~\cite{HornThomasWallmeier08}, which in turn extended results of the third author's dissertation~\cite{Wallmeier08}.

\section{Definitions}
\label{sec_defs}

We denote the set of non-negative integers by $\nats$. For every $k\in \nats$ we
define $\natspref{k} = \set{1, \ldots, k}$, so in particular $\natspref{0} =
\emptyset$. The power set of a set~$S$ is denoted by $2^S$. The last letter of a finite non-empty word $w$ is denoted by $\last(w)$.

An arena $\arena = (V, V_0, V_1, E)$ consists of a finite, directed graph
$(V,E)$, $V_0 \subseteq V$ and $V_1 = V \setminus V_0$, where $V_i$ denotes
the vertices of Player~$i$. In examples, we denote the vertices of Player~$0$ by circles and the vertices of Player~$1$ by squares. We require every vertex to have an outgoing edge to avoid having to deal with finite plays. The size $|\arena|$ of
$\arena$ is the cardinality of~$V$. A play in $\arena$ starting in $v\in V$ is
an infinite sequence $\rho = \rho_0 \rho_1 \rho_2 \cdots$ with $\rho_0 =
v$ and $(\rho_n, \rho_{n+1}) \in E$ for all $n \in \nats$. 

A game~$\game = (
\arena, \win )$ consists of an arena~$\arena$ and a set $\win \subseteq
V^\omega$ of winning plays for Player~$0$, which is often defined implicitly.
The set of winning plays for Player~$1$ is $V^\omega \setminus \win$.

A strategy for Player~$i$ is a mapping $\sigma \colon V^*V_i \rightarrow V$
such that $(v, \sigma(wv)) \in E$ for all $wv \in V^* V_i$. We say that
$\sigma$ is positional if $\sigma(wv) = \sigma(v)$ for every $wv \in V^*V_i$.
A play~$\rho_0 \rho_1 \rho_2 \cdots$ is consistent with $\sigma$ if
$\rho_{n+1} = \sigma( \rho_0 \cdots \rho_n)$ for every~$n$ with $\rho_n \in
V_i$. Given a set~$W \subseteq V$, we denote by $\behavior(W, \sigma)$ the set
of plays that start in $W$ and are consistent with $\sigma$. A
strategy~$\sigma$ for Player~$i$ is a winning strategy from $W$ if every play
in $\behavior(W, \sigma)$ is winning for Player~$i$. The winning region~$W_i(
\game )$ of Player~$i$ in $\game$ contains all vertices from which Player~$i$
has a winning strategy. We always have $W_0( \game ) \cap W_1( \game ) =
\emptyset$ and $\game$ is determined if~$W_0(\game) \cup W_1(\game) = V$. A
winning strategy for Player~$i$ is uniform, if it is winning from
$W_i(\game)$.

A memory structure~$\mem = (M, \init, \update)$ for an arena $(V, V_0, V_1, E
)$ consists of a finite set~$M$ of memory states, an initialization
function~$\init \colon V \rightarrow M$, and an update function~$\update\colon
M\times V\rightarrow M$. The update function can be extended to
$\update^*\colon V^+\rightarrow M$ by defining $\update^*( \rho_0 ) = \init( \rho_0 )$
and $\update^* ( \rho_0 \cdots \rho_n \rho_{n+1} ) = \update( \update^*(
\rho_0 \cdots \rho_n), \rho_{n+1})$. A next-move function (for Player~$i$)
$\nxt \colon V_i \times M \rightarrow V$ has to satisfy $(v, \nxt(v, m)) \in
E$ for all $v \in V_i$ and $m \in M$. The next-move function induces a strategy~$\sigma$ for
Player~$i$ with memory~$\mem$ via the definition $\sigma(\rho_0\cdots\rho_n)=\nxt( \rho_n,
\update^*( \rho_0 \cdots \rho_n))$. The size of $\mem$ (and, slightly abusive, $\sigma$) is $|M|$. A strategy~$\sigma$ is finite-state if it
can be implemented with a memory structure.

An arena $\arena = (V, V_0, V_1, E)$ and a memory structure $\mem = (M, \init,
\update)$ for $\arena$ induce the expanded arena $\arena\times\mem = (V \times
M, V_0 \times M, V_1 \times M, E' )$ where we have $((v,m), (v',m')) \in E'$ if and
only if $(v,v') \in E$ and $\update(m, v' ) = m'$. Furthermore, every play $\rho =
\rho_0\rho_1\rho_2\cdots$ in the original arena~$\arena$ has a unique extended
play~$\extend(\rho) = (\rho_0, m_0) (\rho_1, m_1) (\rho_2, m_2) \cdots$ in
$\arena \times \mem$ defined by $m_0 = \init( \rho_0 )$ and $m_{n+1} =
\update(m_n, \rho_{n+1})$, i.e., we have $m_n = \update^*(\rho_0 \cdots \rho_n)$.
Dually, every play~$\rho = (\rho_0, m_0)(\rho_1, m_1)(\rho_2, m_2)\cdots$ in
$\arena \times \mem$ has a projected play~$\project(\rho) =
\rho_0\rho_1\rho_2\cdots$ in $\arena$. Note that we have
$\project(\extend(\rho)) = \rho$, but $\extend(\project(\rho')) = \rho'$ is
only true if $\rho'$ starts in a vertex of the form~$(v, \init(v))$.

A game $\game = ( \arena, \win)$ is reducible to $\game' = ( \arena', \win')$
via $\mem$, written $\game \le_{ \mem } \game'$, if $\arena' = \arena \times
\mem$ and every play $\rho$ in $\game$ is won by the player who wins the
extended play~$\extend(\rho)$ in $\game'$, i.e., $\rho \in \win$ if and only
if $\extend(\rho) \in \win'$.

\begin{lmm} 
\label{lem_reductiongivesmemory}
Let $\game$ be a game with vertex set $V$ and $W \subseteq V$. If $\game \le_{
\mem } \game'$ and Player~$i$ has a positional winning strategy for $\game'$
from $\{(v, \init(v))\mid v \in W\}$, then she has a winning strategy
with memory~$\mem$ for $\game$ from~$W$.
\end{lmm} 

So in particular, if a player has a uniform positional winning strategy for
 $\game'$, then she has a uniform finite-state winning strategy with memory~$\mem$ for $\game$.
\section{Request-Response Games}
\label{sec_rrgames}
  
A request-response game (RR game for short) is denoted by $(\arena, \rrcond)$ where $\arena$ is an arena and $Q_j$ and $P_j$ are subsets of the set of $\arena$'s vertices. A
vertex in $Q_j$ is referred to as a request of the $j$-th condition, while a vertex in $P_j$ is a response for the $j$-th condition. Intuitively,
Player~$0$'s goal is to answer every request by a later visit to a corresponding response. Formally, a play~$\rho_0 \rho_1 \rho_2 \cdots$ is winning for Player~$0$, if for every $j \in \natspref{k}$ and every $n$, if $\rho_n \in Q_j$, then there exists an $n' \ge n$ such that $\rho_{n'} \in P_j$. We say that a request of condition~$j$ is open after a play prefix~$w$, if $w$ contains a vertex in $Q_j$ that is not followed by a vertex in $P_j$.

\begin{xmpl}
\label{ex_rr}
Consider the RR game in Figure~\ref{fig_rr}. At vertex~$q$ Player~$1$ can request either condition $1$ and/or condition~$2$, while at vertex~$p$, Player~$0$ can either answer condition~$1$ or condition~$2$ or none of them. Alternatingly answering condition~$1$ and $2$ is a uniform winning strategy for Player~$0$ from every vertex.

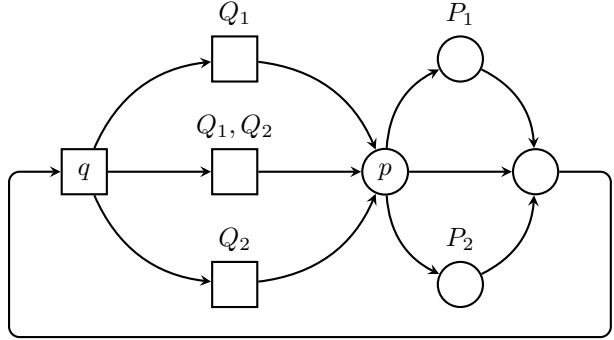
\begin{figure}[h]
	\begin{center}
		\begin{tikzpicture}
		\node[p1] 	at (0,0) 	(s) 	{$q$};
		\node[p1]	at (2,1.5)	(r1)	{};
		\node		at (2,2.1)		{$Q_1$};
		\node[p1]	at (2,-1.5)	(r2)	{};
		\node		at (2,-0.9)		{$Q_2$};
		\node[p1]	at (2,0)	(r12) 	{};
		\node		at (2,.6)		{$Q_1,Q_2$};

		\node[p0]	at (4,0)	(c)		{$p$};
		\node[p0]	at (5,1.5)	(p1)	{};
		\node		at (5,2.1)			{$P_1$};
		\node[p0]	at (5,-1.5) 	(p2)	{};
		\node		at (5,-.9)		{$P_2$};

		\node[p0]	at (6,0)	(e)	{};

		\path
		(s) edge[bend left] (r1)
		(s) edge[bend right] (r2)
		(s) edge(r12)
		(r1) edge[bend left] (c)
		(r2) edge[bend right] (c)
		(r12) edge (c)
		(c) edge[bend left] (p1)
		(c) edge[bend right] (p2)
		(c) edge (e)
		(p1) edge[bend left] (e)
		(p2) edge[bend right] (e);

		\path[draw, thick, rounded corners, -stealth]
		(e.east) -- (7, 0) -- (7,-2.2) -- (-1, -2.2) -- (-1, 0) -- (s.west); 

		\end{tikzpicture}
\end{center}
\caption{The RR game for Example~\ref{ex_rr} and Example~\ref{ex_values}.}
\label{fig_rr}
\end{figure}
\end{xmpl}

There is an intuitive notion of the waiting time between a request and its
(earliest) response, which we formalize in the following. The waiting times
are then aggregated to measure the quality of a play and the quality of a
strategy (both from Player~$0$'s point of view). First, we define the
waiting time for condition~$j$, denoted by $\waittime_j \colon V^* \rightarrow
\nats$, inductively via $\waittime_j(\epsilon) = 0$, and
\[ \waittime_j(wv) = \begin{cases}
0 		&\text{if $\waittime_j(w) = 0$ and $v \notin Q_j \setminus P_j$,}\\
1 		&\text{if $\waittime_j(w) = 0$ and $v \in Q_j \setminus P_j$,}\\
0 		&\text{if $\waittime_j(w) > 0$ and $v \in P_j$,}\\
\waittime_j(w)+1 
		&\text{if $\waittime_j(w) > 0$ and $v \notin P_j$.}\\
\end{cases} \] 

Note that while a request of condition~$j$ is open, additional requests are
ignored, i.e., we are only interested in the waiting time of the earliest
request that is open, but not in the number of requests (of a single
condition) that are open. In \cite{Zimmermann09} an extension of RR games is investigated, where the waiting times take the number of open requests into account as well.  

\begin{rmrk}
\label{rem_wtevolution}
If $\waittime_j(x) \le \waittime_j(y)$, then $\waittime_j(xz) \le
\waittime_j(yz)$ for every $z \in V^*$.
\end{rmrk}

We summarize the waiting times of a play prefix~$w$ in its waiting time vector~$\waitvec(w) = ( \waittime_1(w), \ldots, \waittime_k(w)) \in \nats^k$ and compare such vectors component-wise, i.e.,
$\waitvec(x) \le \waitvec(y)$ if $\waittime_j(x) \le \waittime_j(y)$ for every
$j$.

We say that a strategy~$\sigma$ for Player~$0$ (uniformly) bounds the waiting times for condition~$j$ by $b \in \nats$, if every play prefix~$w$ that starts in $W_0(\game)$ and is consistent with $\sigma$ satisfies $\waittime_j(w) \le b$. If $\sigma$ bounds the waiting times for every condition, then it is a uniform winning strategy.

Now, we use the waiting times to define the quality of plays and strategies from Player~$0$'s point of view: for every $j$ we fix a strictly increasing penalty function~$f_j\colon \nats
\rightarrow \nats$ (which implies that $f_j$ is unbounded) and define the penalty of a play prefix~$w$ for the $j$-th condition
by $\pnlty_j(w) = f_j(\waittime_j(w))$ and the overall penalty of $w$ by
$\pnlty(w) = \sum_{j\in\natspref{k}} \pnlty_j(w)$. We aggregate the penalties
of an infinite play~$\rho$ to the value of this play by taking the limit superior of
the mean accumulated penalties, i.e., we define
\begin{equation*}\val(\rho) = \limsup_{n \rightarrow \infty} \frac{1}{n}\sum_{\ell = 0}^{n-1} 
 \pnlty(\rho_0 \cdots \rho_{\ell}) \enspace.\label{eq_playval}\end{equation*}
Finally, the value of a strategy~$\sigma$ from a vertex~$v$ is 
\[\val(\sigma, v) = \sup\nolimits_{\rho\in\behavior(v, \sigma)} \val(\rho) 
\enspace.\]
Note that we do not parameterize the functions~$\pnlty_j$, $\pnlty$, and $\val$ with the penalty functions~$f_j$, although they depend on them. This is done to improve readability. In the following, we will always ensure that the penalty functions are clear from the context.

\begin{xmpl}
\label{ex_values}
Using the identity function as penalty functions~$f_j$, the uniform winning strategy described in Example~\ref{ex_rr} has value~$\frac{56}{10} $ from every vertex, which is witnessed by Player~$1$ always requesting both conditions every time  when at vertex~$q$. Every play consistent with this strategy and the alternating-response strategy for Player~$0$ ends up in a loop of length~$10$, in which the sum of the waiting times (which are also the penalties) is 56. The value of this play is equal to the length of the loop divided by its length, hence $\frac{56}{10} $. Every other play has a smaller or equal value. Thus, the value of the strategy is also equal to $\frac{56}{10} $, independently of the initial vertex.
\end{xmpl}

It is important to note that we still consider the game as a zero-sum one; we
just associate values to plays and strategies and are interested in optimal
strategies, i.e., a winning strategy~$\sigma$ such that
every other winning strategy~$\sigma'$ satisfies $\val(\sigma',v) \ge
\val(\sigma, v)$ for every vertex~$v$. Note that it is a priori not even clear
whether an optimal strategy exists. 

The sum of penalties $ \sum_{\ell = n}^{n+n'} \pnlty(\rho_0 \cdots \rho_{\ell})$
for a play infix~$\rho_n \cdots \rho_{n+n'}$ with an open request grows (at least) quadratically in $n'$, since the penalty functions~$f_j$ are strictly increasing. Our result on the existence of optimal finite-state strategies relies on this growth, as evidenced by the following example, which shows that optimal finite-state strategies do not necessarily exist if we allow constant penalty functions.
 
\begin{xmpl}
\label{ex_nooptstrat}
Assume we use constant penalty functions~(e.g., $f_j(0)=0$ and $f_j(n) = 1$ for every $n > 0$) to measure the quality of plays and consider the RR game depicted in Figure~\ref{fig_nooptstrat}. Player~$0$ wins from every vertex by traversing both loops infinitely often, which is also necessary to win.

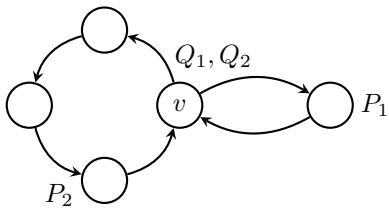
\begin{figure}[h]
	\begin{center}
		\begin{tikzpicture}
			
		\node[p0]		at (0,0) 	(c) 	{$v$};
		\node			at (.45,.65)			{$Q_1, Q_2$};
		\node[p0]		at (2,0) 	(r) 	{};
		\node			at (2.6,0)			{$P_1$};
		\node[p0]		at (-1,1) 	(lu) 	{};
		\node[p0]		at (-2,0) 	(l) 	{};
		\node[p0]		at (-1,-1)	(ld) 	{};
		\node			at (-1.6,-1.2)		{$P_2$};		
		\path
		(c) edge[bend left]		(r)
		(r) edge[bend left]		(c)
		(c) edge[bend right]	(lu)
		(lu)edge[bend right]	(l)
		(l) edge[bend right]	(ld)
		(ld)edge[bend right]	(c)	; 
	 
		\end{tikzpicture}
\end{center}
\caption{The RR game for Example~\ref{ex_nooptstrat}.}
\label{fig_nooptstrat}
\end{figure}
In the following, we only consider plays starting in $v$. As it is Player~$0$'s turn at every vertex, we can identify strategies and plays (and their values are equal). 

If a strategy is finite-state, then its play is ultimately periodic, i.e., of the form $\rho = \rho_0 \cdots \rho_{m-1}(\rho_m \cdots \rho_n)^\omega$, where we assume w.l.o.g.\ $\rho_m = v$. Then, $\val(\rho)$ is equal to the number of positions of the period~$\rho_m \cdots \rho_n$ where condition~$1$ is open plus the number of positions of $\rho_m \cdots \rho_n$ where condition~$1$ is open. 

Now, consider the play infix~$\rho_m' \cdots \rho_n'$ obtained from the period by replacing every visit to the left loop by two visits to the right loop. The value of the play $\rho' = \rho_0 \cdots \rho_{m-1}(\rho_m \cdots \rho_n \rho_m' \cdots \rho_n')^\omega$ is strictly smaller than the value of $\rho$, since visiting the left loop is more costly than visiting the right one twice.

As $\rho'$ can also be generated by a finite-state strategy,  we have shown  that there is no optimal finite-state strategy when considering constant penalty functions.

%
\end{xmpl}

The values~$\val(\rho)$ and $\val(\sigma, v)$ measure the quality of a play and a strategy from Player~$0$'s point of view. However, it is not true that a play (a strategy) is winning for Player~$0$ if, and only if, it has a finite value. One direction holds, as claimed in the next lemma, while the other one can shown to be false by
considering a play in which Player~$0$ allows herself more and more time to
answer the requests.

\begin{lmm}
\label{lemma_finitevalue}
Let $v$ be a vertex, let $\rho$ be a play, and let $\sigma$ be a strategy for Player~$0$.
\begin{enumerate}
\item\label{lemma_finitevalue_play}
If $\val(\rho) < \infty$, then $\rho$ is a winning play for Player~$0$.

\item\label{lemma_finitevalue_strategy}
If $\val(\sigma, v) < \infty$, then $\sigma$ is a winning strategy for
Player~$0$ from~$v$.

\item\label{lemma_finitevalue_region}
If $v \in W_1(\game)$, then $\val(\sigma, v) = \infty$.

\end{enumerate}
\end{lmm}

\begin{proof}
(\ref{lemma_finitevalue_play}) 
Consider the contraposition: let $\rho=\rho_0\rho_1\rho_2\cdots$ be winning
for Player~$1$. Then, some condition $j$ is requested at some position~$n$,
but never answered afterwards. Thus, 
$\pnlty_j(\rho_0\cdots\rho_{n+n'}) \ge
f_j(n') \ge n'$ for every $n'$ (recall that $f_j$ is strictly increasing) and
therefore
\[\frac{1}{n+n'}\sum_{\ell=0}^{n+n'-1}\pnlty(\rho_0\cdots\rho_\ell) \ge
\frac{1}{n+n'} \frac{n' (n'-1)}{2} =
\frac{n'-1}{2 \left(\frac{n}{n'}+1\right)}\] 
for all $n'$, which diverges to infinity when $n'$ tends to infinity. Thus, $\val(\rho) = \infty$.

(\ref{lemma_finitevalue_strategy}) 
Again, we consider the contraposition: let $\sigma$ not be a winning strategy
from $v$. Then, there exists a play~$\rho \in \behavior(v, \sigma)$ that is
winning for Player~$1$. Thus, $\val(\rho) = \infty$ as we have just shown and $\val(\sigma, v)=\infty$, too.

(\ref{lemma_finitevalue_region})
Let $\tau$ be a winning strategy for Player~$1$ from $v$ and consider the
unique play $\rho$ that starts in $v$ and is consistent with $\sigma$ and
$\tau$. We have $\val(\rho) = \infty$, as $\tau$ is a winning strategy for
Player~$1$, and therefore $\val(\sigma, v) = \infty$, as $\rho$ is also
consistent with $\sigma$ and starts in $v$.
\end{proof}

To conclude this introductory section on RR games, we recall the proof of finite-state determinacy of
RR games, which proceeds by a reduction to B\"{u}chi games. The strategy
obtained by this reduction yields a first upper bound on the value of an
optimal strategy in an RR game.

The winning condition of a B\"{u}chi game is a set~$F$ of vertices and
Player~$0$ wins a play if it visits $F$ infinitely often. Alternatively, one
can define a B\"{u}chi game to be an RR game with a single condition of the
form~$(V\setminus F, F)$ which is
satisfied if and only if $F$ is visited infinitely often. As B\"{u}chi games
are positionally determined, such a reduction suffices to prove the
following result.

\begin{thrm}[\hspace{-.00001cm}\cite{WallmeierHuttenThomas03}]
\label{theorem_rrdeterminacy}
RR games are determined with finite-state strategies of size~$k
2^{k+1}$, where $k$ denotes the number of RR conditions.
\end{thrm}

\begin{proof}
Let $\game = (\arena, \rrcond)$ be an RR game with (w.l.o.g.) $k > 1$
conditions. Consider the memory structure~$\mem = (M, \init, \update)$ with $M
= \pow{\natspref{k}} \times \natspref{k} \times \set{0,1}$, $ \init(v) =
(\set{j \mid v \in Q_j \setminus P_j}, 1, 0)$, and $\update(R, c, f) = (R',
c', f')$ where
\begin{itemize}

\item  $R' = \left( R \cup \set{j \mid v \in Q_j } \right) \setminus \set{j 
\mid v \in P_j }$,

\item $c' = c$ if $c \in R \cap R'$, and $c' = (c \mod k) + 1$ otherwise,
and

\item $f = 1$ if $c' \not= c$, and $f=0$, otherwise.

\end{itemize}
So, $R$ keeps track of open requests, $c$ is a cyclic counter over
$\natspref{k}$ that is incremented every time its current value is not an open
request, and the flag~$f$ is equal to $1$ if and only if $c$ has changed its
value. So, there is an unanswered request if and only if $f$ is equal to $0$
from some point onwards. Thus, consider the set~$F = V \times
(\pow{\natspref{k}} \times \natspref{k} \times \set{1})$. Then, we have $\game
\le_\mem (\arena\times\mem, F)$, i.e., the result follows from
Lemma~\ref{lem_reductiongivesmemory}.
\end{proof}

The upper bound on the memory requirements was slightly lowered and (almost)
matching lower bounds were proven in~\cite{ChatterjeeHenzingerHorn11}.
However, for our purposes, the simple bound presented here is sufficient to obtain
an upper bound on the value of an optimal strategy.

\begin{crllr} \label{corollary_upperboundonval} In every RR game~$\game$,
Player~$0$ has a winning strategy $\sigma$ with $\val(\sigma, v) \le \sum_{j
\in \natspref{k}}f_j(s k  2^k)$ for every $v \in W_0(\game)$, where
$s$ denotes the size of the arena and $k$ the number of RR conditions.
\end{crllr}

\begin{proof} 
Let $\mem$ be the memory structure defined in the
proof of Lemma~\ref{theorem_rrdeterminacy} and let $\sigma$ be a uniform winning strategy~$\sigma$ for $\game$ with memory $\mem$. We show that $\sigma$ has the desired properties. To this end, we prove that $\sigma$ bounds the waiting time of every condition by $s k 2^k$. Then, we have 
\begin{align*}\val(\rho) =& 
\limsup_{n \rightarrow \infty} \frac{1}{n}\sum_{\ell = 0}^{n-1}
\sum_{j \in \natspref{k}} f_j ( \waittime_j(\rho_0 \cdots \rho_{\ell}))\\
\le &\limsup_{n \rightarrow \infty} \frac{1}{n}\sum_{\ell = 0}^{n-1}
\sum_{j \in \natspref{k}} f_j ( s k 2^k ) \le
\sum_{j \in \natspref{k}} f_j( s k 2^k) \enspace.
\end{align*}
for all $\rho$ starting in $W_0(\game)$ that are consistent with $\sigma$, which implies our claim.

Towards a contradiction, assume we have $\waittime_j(\rho_0 \cdots \rho_n) > s
 k  2^k$. Then, $j$ is in the first component (which keeps track of
open requests) of the last $s k  2^k +1$ memory states reached
during $\rho_0 \cdots \rho_n$. But there are only $k 2^k$ memory states
that contain $j$ in the first component. Hence, there are positions $m < m'$
in this interval such that $\rho_m = \rho_{m'}$ and $\update^* (\rho_0 \cdots
\rho_m) = \update^* (\rho_0 \cdots \rho_{m'})$.

Now, consider the play $\rho' = \rho_0 \cdots \rho_{m-1} (\rho_{m} \cdots
\rho_{m'-1})^\omega$ obtained by repeating the loop between positions~$m$ and $m'$, which is also in $\behavior(W_0(\game), \sigma)$. But $\rho'$
contains an unanswered request, since condition~$j$ is open at $\rho_m$ and never answered during $\rho_m \cdots \rho_{m'-1}$. This contradicts the fact that $\sigma$ is winning from $W_0(\game)$.
\end{proof}

The exponential upper bound on the waiting times used in the proof of Corollary~\ref{corollary_upperboundonval} gives a correction to a claim of~\cite{HornThomasWallmeier08, Zimmermann09} where the bound~$sk$ is used. Next, we give an example showing a matching exponential lower bound.

\begin{xmpl}
Consider the RR game depicted in Figure~\ref{fig_explowerboundonwt}. We consider plays starting at vertex~$i$, where all four RR conditions are requested. From there, the only move leads to the hub vertex~$h$, where Player~$0$ has to move into one of the four blades, each one of them associated with one of the RR conditions. The first vertex of the blade for condition~$j \in [4]$ ($c_j$ in the figure) is in $P_j$, i.e., condition~$j$ is responded to. From this vertex, Player~$1$ can either move to a sink vertex (called $s_j$) where every condition with index larger than $j$ is answered, too, or he can move to a vertex (called $v_j$) where all conditions with index smaller than $j$ are requested again. From this vertex, the only move leads back to the hub. Due to the existence of the former move, Player~$1$ can win if the $j$-th blade is entered while a request of a condition with smaller index is open. 

\begin{figure}[h]
\begin{center}
\begin{tikzpicture}

\begin{scope}
\node[p0] 				at (0,0) 	(c) 	{$h$};

\node[p0]				at (0,-4)	(i)		{$i$};
\node[anchor = west]	at (.3,-4)	() 		{$Q_1, Q_2, Q_3, Q_4$};

\node[p1]				at (1,1)	(11)	{$c_1$};
\node[anchor = east]	at (.7,1)	()		{$P_1$};
\node[p1]				at (2,2)	(12)	{$s_1$};
\node[anchor = east]	at (1.7,2)	()		{$P_2, P_3, P_4$};
\node[p1]				at (3,3)	(13)	{$v_1$};
\node[anchor = east]	at (2.7,3)	()		{$ $};

\node[p1]				at (1,-1)	(21)	{$c_2$};
\node[anchor = west]	at (1.3,-1)	()		{$P_2$};
\node[p1]				at (2,-2)	(22)	{$s_2$};
\node[anchor = west]	at (2.3,-2)	()		{$P_3, P_4$};
\node[p1]				at (3,-3)	(23)	{$v_2$};
\node[anchor = west]	at (3.3,-3)	()		{$Q_1$};

\node[p1]				at (-1,-1)	(31)	{$c_3$};
\node[anchor = west]	at (-.7,-1)	()		{$P_3$};
\node[p1]				at (-2,-2)	(32)	{$s_3$};
\node[anchor = west]	at (-1.7,-2)()		{$P_4$};
\node[p1]				at (-3,-3)	(33)	{$v_3$};
\node[anchor = west]	at (-2.7,-3)()		{$Q_1, Q_2$};

\node[p1]				at (-1,1)	(41)	{$c_4$};
\node[anchor = east]	at (-1.3,1)	()		{$P_4$};
\node[p1]				at (-2,2)	(42)	{$s_4$};
\node[anchor = east]	at (-2.3,2)	()		{$ $};
\node[p1]				at (-3,3)	(43)	{$v_4$};
\node[anchor = east]	at (-3.3,3)	()		{$Q_1, Q_2, Q_3$};

\end{scope}

\path[thick]
(i)			edge				(c)
(c)			edge				(11)	
(11)		edge				(12)
(11)		edge[bend right]	(13)
(13.south)	edge[bend left]		(c)
(12)		edge[loop above]	()

(c)			edge				(21)
(21)		edge				(22)
(21)		edge[bend right]	(23)
(23.west)	edge[bend left]		(c)
(22)		edge[loop above]	()

(c)			edge				(31)
(31)		edge				(32)
(31)		edge[bend right]	(33)
(33.north)	edge[bend left]		(c)
(32)		edge[loop below]	()

(c)			edge				(41)
(41)		edge				(42)	
(41)		edge[bend right]	(43)
(43.east)	edge[bend left]		(c)
(42)		edge[loop below]	();

\end{tikzpicture}
\caption{An RR game with exponential waiting times.}
\label{fig_explowerboundonwt}
\end{center}
\end{figure}

Player~$0$ has a winning strategy for this RR game from vertex~$i$ by always moving to the blade of the smallest open condition. This strategy takes $2^4-1$ visits to the hub to answer the request of condition~$4$, since every smaller condition is requested after answering the smallest open condition. Once all requests are answered, Player~$0$ can always move to the first blade, which does not generate new requests. Deviating from this strategy either generates additional requests (if moving to a blade of a condition with smaller index than the currently smallest open one) and thereby prolongs the time it takes to answer condition~$4$, or allows Player~$1$ to move to a sink vertex where he wins (if moving to a blade of a condition with larger index than the currently smallest open one). Thus, every winning strategy from vertex~$i$ for Player~$0$ in this game takes at least $2^4-1$ visits to the hub before condition~$4$ is answered.

This game can be generalized by having $k$ conditions and $k$ blades. Then, it takes at least $2^k-1$ visits to the hub to answer the request of condition~$k$. Thus, the waiting time for condition~$k$ is larger than $2^k$ before it is eventually answered. 
\end{xmpl}
\section{Bounding the Waiting Times in RR Games}
\label{sec_bounding}

In this section, we show that for every strategy whose value is small from every vertex in $W_0(\game)$ there is a strategy with smaller or equal values that additionally bounds the waiting times by some bound~$b$, which only depends on the size of the arena and the number of RR conditions. This restricts the search space for optimal strategies to a finite one (in terms of possible waiting time vectors). In the next section, this space is turned into an arena for a mean-payoff game. Intuitively, the arena tracks plays of the RR game and their waiting times up to the threshold~$b$. The value of a play in the mean-payoff game is the value of the tracked play in the original game. Thus, an optimal strategy for Player~$0$ in the mean-payoff game (which can be effectively computed) can be turned into an optimal winning strategy for Player~$0$ in the RR game.

We proceed as follows: in Subsection~\ref{subsec_dickson}, we derive the bound~$b$ and in Subsection~\ref{subsec_bounding} we show that we can turn every strategy of small value into a strategy of smaller or equal value whose waiting times are bounded by $b$.

\subsection{Dickson's Lemma for Waiting Times}
\label{subsec_dickson}

Given a strategy~$\sigma$ with small values for $W_0(\game)$ we need to construct a strategy with smaller or equal values that also bounds the waiting times by a constant that only depends on the number of vertices and RR conditions and the penalty functions. We achieve this by removing loops from plays in which the waiting time is high for some condition. However, this might have an effect on the waiting times for other conditions as well: in the worst case we might remove an answer to a request, thereby increasing the waiting time or even generating a losing play. To avoid this, we only remove a loop if the waiting time vector at the end of the loop is larger than at the beginning. 

This removal process is iterated until ad infinitum, i.e., in the limit there are no more such loops. Hence, our bound~$\bnd$ has to be an upper bound on the length of play infixes without such a loop. We derive $\bnd$ in this subsection by giving a quantitative version of Dickson's Lemma~\cite{Dickson13}, which states that there is no infinite play prefix without such a loop. However, the lemma does not give an explicit bound on the length of a play without such a loop. Indeed, if we allow arbitrary vectors of natural numbers (this is the setting of Dickson's Lemma), there are arbitrarily long sequences. But by exploiting the simple update-rule of the waiting times -- increment or reset -- we are able to obtain a doubly-exponential bound~$\bnd$.

Let $\rho$ be a play of the RR game~$\game$. We say that a pair of positions~$(n_1,n_2)$ of $\rho$ with $n_1 <  n_2$ is dickson, if we have $\rho_{n_1} = \rho_{n_2}$ and $\waitvec(\rho_0 \cdots \rho_{n_1}) \le \waitvec(\rho_0 \cdots \rho_{n_2})$. Note that the notion is defined with respect to the whole play prefix~$\rho_0 \cdots \rho_{n_2}$, since the waiting times are computed starting at the first position of $\rho$. An infix is dickson, if it contains a dickson pair of positions, otherwise it is non-dickson.

The goal of this section is to define a function~$\bnd \colon \nats \times \nats \rightarrow \nats$ such the following is true for every RR game with $s$ vertices and $k$ RR conditions: every play infix of length at least $\bnd(s,k)$ has a dickson pair. Note that this implies that we have to deal with arbitrarily high waiting times at the beginning of the infix. 

We define $\bnd$ by induction over $k$, the number of RR conditions. For $k=0$, we  have $\bnd(s, k) = s + 1$, since every state repetition yields a
dickson pair. Now, consider a game with $k > 0$ RR conditions. We begin by stating a lemma that restricts the combinations of
values that can appear in waiting time vectors in a non-dickson infix: intuitively, not too many waiting times can be large at the same time, since this would imply the existence of a dickson pair in the remaining conditions.  This is also a dickson pair for all conditions, since the large values only increase between these two positions.

\begin{lmm}
Let $\rho_m \cdots \rho_{m+\ell}$ be a non-dickson play infix of a play $\rho$ in an RR game with $s$ vertices and $k$ RR conditions. For
every $j$ in the range $0 \le j \le k-1$ and every $n$ in the range $\bnd(s,
k-(j+1)) \le n \le \ell$, $\waitvec(\rho_0 \cdots \rho_{m+n})$ contains at most $j$ entries that are larger than $\bnd(s, k-(j+1))$.
\end{lmm} 

\begin{proof}
\label{lem_valuesinnondicksoninfix}
Towards a contradiction, assume there is a $j$ such that $\waitvec(\rho_0 \cdots \rho_{m+n})$ contains $j+1$ entries that are larger than $\bnd(s, k-(j+1))$, where $\bnd(s, k-(j+1)) \le n \le m + \ell$. We denote the set of coordinates of these entries in the waiting time vectors by $J$, i.e., $J \subseteq \natspref{k}$. The entries at the coordinates in $J$ are updated by increasing them during the last $\bnd(s, k-(j+1))$ positions before $m+n$, which are all positions contained in the infix. 

Now, consider the projection to the $k-(j+1)$ coordinates not in $J$: there is a dickson pair in the infix~$\rho_{m + n + 1 -\bnd(s, k-(j+1))} \cdots \rho_{m +n}$, as it has length~$\bnd(s, k-(j+1))$. This is also a dickson pair when considering all coordinates, since the values at the coordinates in $J$ are strictly increasing during this infix. This yields the desired contradiction.
\end{proof}	
	
Due to the previous lemma, after $\bnd(s, k-1)$ positions in a non-dickson infix,
every vector has no entry larger than $\bnd(s,k-1)$, at most one entry larger than
$\bnd(s, k-2)$, at most two entries larger than $\bnd(s, k-3)$, and in general, at
most $j$ entries larger than $d(s, k-(j+1))$ for every $j \in \set {0, \ldots, k-1}$.
Rephrasing this, we obtain that every such vector contains an entry smaller than
$\bnd(s, 0)$, another entry smaller than $\bnd(s, 1)$, another entry smaller than
$\bnd(s, 2)$, and so on. The number of such vectors is bounded by $k!  \prod_{j=0}^{k-1}
\bnd(s,j)$. So, we can define for $k>0$ \[  \bnd(s, k)  = \left( \bnd(s,k-1) +  s  k! 
\prod_{j=0}^{k-1} \bnd(s,j)\right) +1 \enspace.\]
The first summand is due to the fact that the bounds only hold after $\bnd(s,k-1)$
steps, and the factor~$s$ in the second summand takes account of the fact that we need a state repetition in a dickson-pair as well.

\begin{lmm}
Let $\game$ be an RR game with $s$ vertices and $k$ RR conditions.
\begin{enumerate}
\item\label{lem_bound_correct} Every play infix of length~$b(s,k)$ has a dickson pair.
\item\label{lem_bound_ub} We have $ b(s,k) \in \mathcal{O}(2^{2^{s \cdot k + 2}}) $.
\end{enumerate}
\end{lmm}

\begin{proof}(\ref{lem_bound_correct}) This follows directly from Lemma~\ref{lem_valuesinnondicksoninfix} and from the arguments presented after it.

(\ref{lem_bound_ub}) 
We show $\bnd(s,k) \le 2^{2^{k-1}} (s+1)^{2^k} k! \prod_{j=1}^{k-1}(j!)^{2^{k-(j+1)}}$  for every $k>0$, which implies the claim. 
Throughout the proof we use the following upper bound 
\[ \bnd(s, k)  = \left( \bnd(s,k-1) +  s  k! 
\prod_{j=0}^{k-1} \bnd(s,j)\right) +1 
\le 2 (s+1) k! \prod_{j=0}^{k-1} \bnd(s,j)
\enspace\]
for $k>0$, which also proves our claim for $k=1$. Now, consider $k>1$. We have 
\begin{align*}
b(s,k) & \le 2 (s+1) k! \prod_{j=0}^{k-1} \bnd(s,j) \\
& = 2 (s+1)^2 k! \prod_{j=1}^{k-1} \bnd(s,j) \\
& \le  2 (s+1)^2 k! \prod_{j=1}^{k-1} \left( 2^{2^{j-1}} (s+1)^{2^j} j! \left( \prod_{j'=1}^{j-1}(j'!)^{2^{j-(j'+1)}} \right)\right)\\
& =  2^{1+\sum_{j=1}^{k-1} {2^{j-1}}} \cdot (s+1)^{2+\sum_{j=1}^{k-1} {2^{j}}} \cdot k! \prod_{j=1}^{k-1} j! \prod_{j'=1}^{j-1}(j'!)^{2^{j-(j'+1)}}\\
& =  2^{1+\sum_{j=0}^{k-2} {2^{j}}} \cdot (s+1)^{1+\sum_{j=0}^{k-1} {2^{j}}} \cdot  k! \prod_{j=1}^{k-1}(j!)^{2^{k-(j+1)}}\\
& =  2^{2^{k-1}} \cdot (s+1)^{2^k} \cdot k! \prod_{j=1}^{k-1}(j!)^{2^{k-(j+1)}}\enspace.\\
\end{align*}
The equality $ \prod_{j=1}^{k-1} j! \prod_{j'=1}^{j-1}(j'!)^{2^{j-(j'+1)}} = \prod_{j=1}^{k-1}(j!)^{2^{k-(j+1)}}$ used in the second-to-last equality can be shown by a straightforward induction.
\end{proof}

Recently, Czerwi\'nski  et al.\ complemented our doubly-exponential upper bound by proving a doubly-exponential lower bound of $2^{2^{k/2}}$~\cite{CzerwinskiGK14}.
\subsection{Strategies with Small Values and Bounded Waiting Times}
\label{subsec_bounding}

In this subsection, we show how to turn a strategy with small values from every vertex in $W_0(\game)$ into a strategy of smaller or equal values whose waiting times are bounded. To this end, we remove loops of plays in which the waiting time for some condition~$j$ is large. By doing this ad infinitum, we obtain a limit strategy with the desired properties.

Throughout this subsection, we fix an RR game 
$\game = (\arena, \rrcond)$ with $\arena = (V, V_0, V_1, E)$ and $|V| = s$ as well as a penalty function~$f_j$ for every condition~$j$. The goal of this section is to prove the following lemma, which shows that for every strategy of
small value there is a strategy of smaller or equal value that additionally bounds the waiting times. In particular, the result applies to the uniform winning strategy for Player~$0$ from Corollary~\ref{corollary_upperboundonval} which satisfies 
\[\val(\sigma, v) \le \sum_{j \in \natspref{k}}f_j(s k 2^k) =\mathrel{\mathop:} \valgame\] for every $v \in W_0(\game)$.

\begin{lmm}
\label{lemma_bounding}
Let $\sigma$ be a strategy such that $\val(\sigma, v) \le \valgame$ for
every $v \in W_0(\game)$. There is a strategy~$\sigma'$ with $\val(\sigma', v)
\le \val(\sigma, v)$ for every $v \in V$ that uniformly bounds the waiting times for every condition~$j$ by $f^{-1}_j(\valgame) + \bnd(s,k-1)$.
\end{lmm}

Note that $\val(\sigma', v) \le \val(\sigma, v) \le \valgame$ for
every $v \in W_0(\game)$ implies that $\sigma$ and $\sigma'$ are uniform winning strategies.

In this subsection, it is convenient to view a strategy as the set of play prefixes that are consistent with it. This representation simplifies the process of removing loops from the plays which are consistent with the strategy. Also, we only consider plays starting in $W_0(\game)$ since we need to bound the waiting times for such plays, the waiting times for plays starting in $W_1(\game)$ cannot be bounded and are ignored.

Formally, a strategy tree is a prefix-closed language~$t \subseteq V^*$ such that the following conditions are satisfied: 
\begin{enumerate}

\item For every $w_0 \cdots w_n \in t$ we have $w_0 \in W_0(\game) $ and $(w_{n'}, w_{n'+1}) \in E$ for every $n' < n$ (only play prefixes starting in $W_0(\game)$ are in $t$).

\item $t \cap V = W_0(\game)$ (every initial vertex from $W_0(\game)$ is in $t$, but no initial vertex from $W_1(\game)$).

\item For every $w_0 \cdots w_n \in t$ with $w_n \in V_0$ there is a unique $v \in V$ such that $w_0 \cdots w_n v \in t$ (there is a unique successor in $t$ for play prefixes ending in $V_0$).

\item For every $w_0 \cdots w_n \in t$ with $w_n \in V_1$ and every successor~$v$ of $w_n$ in $\arena$ we have $w_0 \cdots w_n v \in t$ (all successors are in $t$ for play prefixes ending in $V_1$).

\end{enumerate}

Every strategy~$\sigma$ can be turned into a strategy tree~$t(\sigma)$ containing exactly the prefixes (including the empty prefix~$\epsilon$) of plays that start in $W_0(\game)$ and are consistent with $\sigma$. Vice versa, every strategy tree~$t$ defines a strategy~$\sigma(t)$ mapping $w$ ending in $V_0$ to the unique vertex~$v$ with $wv \in t$. Note that this strategy is only defined for play prefixes starting in Player~$0$'s winning region that are consistent with $\sigma$. However, this is sufficient for our purposes, since $\sigma(t)(w)$ can be defined arbitrarily for every other $w$.

Given a strategy tree~$t$ and $w \in t$, define \[S_j(t, w) = \set{w w' \in t \mid w' \in (V \setminus P_j)^*} \]
to be the set of continuations of $w$ in which no vertex from $P_j$ is visited, i.e., condition~$j$ is not responded to.

\begin{rmrk}
Let $t$ be a strategy tree such that $\sigma(t)$ is a uniform winning strategy, and let $w \in t$ with $\waittime_j(w) > 0$. Then, $S_j(t,w)$ is finite.
\end{rmrk}

\begin{proof}
Assume $S_j(t,w)$ is infinite. Then, K\"{o}nig's Lemma implies the existence of an infinite play~$w\rho$ whose prefixes are all in $t$, in which a request of condition~$j$ is open after $w$ (as the waiting time is non-zero), but $\rho$ contains no answer. As such a play is consistent with $\sigma(t)$ this contradicts the fact that $\sigma(t)$ is a winning strategy.
\end{proof}

We now formalize the removal of loops, which turns a strategy tree $t$ into a new one denoted by $t'$. Fix some condition~$j \in \natspref{k}$ for which we want to remove loops with large waiting times and fix a strategy tree~$t$ such that $\val(\sigma(t), v) \le \valgame$ for every $v \in W_0(\game)$. Next, we define the tree~$t' \subseteq V^*$ and a mapping~$h \colon t' \rightarrow t$ satisfying $\last(h(w)) = \last(w)$ for every $w \in t' \setminus \set{\epsilon}$. 

We have $\epsilon \in t'$ and $W_0(\game) \subseteq t'$ and define $h(\epsilon) = \epsilon$ and $h(v) = v$ for every $v \in W_0(\game)$. Now, consider some $w \in t'$: we have $h(w) \in t$ with $\last(h(w)) = \last(w)$. For every $h(w)v \in t$ we add $wv$ to $t'$ and it remains to define $h(wv)$. Here, we consider two cases:
\begin{enumerate}
\item If $\waittime_j(h(w)v) \le f^{-1}_j(\val_\game)$, then $h(wv) = h(w)v$.
\item If $\waittime_j(h(w)v) > f^{-1}_j(\val_\game)$, then consider the set~$S_j(t, h(w)v)$. As it is finite we can pick a longest element\footnote{Using the lexicographic order w.r.t.\ some fixed ordering of $V$ to break ties.}~$x$ from $S_j(t, h(w)v)$ satisfying
$ \waitvec(h(w)v) \le \waitvec(x) $ and $\last(x) = v$. Such an element always exists, since $h(w)v \in S_j(t, h(w)v)$ satisfies both requirements. We define $h(wv) = x$.
\end{enumerate}
Note that the property $\last(h(w)) = \last(w)$ is satisfied in both cases of the definition. We begin by listing some straightforward properties of the function~$h$ we use to show that $t'$ is also a strategy tree.  

\begin{rmrk}
\label{remark_cuttingprop}
Let $w = w_0 \cdots w_n \in t'$.
\begin{enumerate}
	
\item\label{remark_cuttingprop_subplay}
$h(w) = w_0 s_0 w_1 s_1 \cdots s_{n-1} w_n$ for some $s_0, \ldots, s_{n-1} \in V^*$.

\item\label{remark_cuttingprop_prefix}
$h(w_0 \cdots w_{n'})$ is a proper prefix of $h(w_0 \cdots w_{n})$ for every $n' < n$. 	

\item\label{remark_cuttingprop_injective}
$h$ is injective.
	
\item\label{remark_cuttingprop_prefixclosure}
If $h(w) = w$, then $h(w_0 \cdots w_{n'}) = w_0 \cdots w_{n'}$ for every $n' < n$.

\item\label{remark_cuttingprop_largewaittime}
Let $w' \in t$. If there is no $w \in t'$ with $h(w) = w'$ then $\waittime_j(w') > f^{-1}_j(\val_\game)$.

\end{enumerate}
\end{rmrk}

Now, we prove that $t'$ is a strategy tree if $t$ is one, and that transforming $t$ into $t'$ preserves waiting time bounds and does not increase the values of the strategy.

\begin{lmm}
\label{lemma_cuttingproponce}
Let $t$ be a strategy tree such that $\val(\sigma(t) ,v ) \le \val_\game$ for every $v \in W_0(\game)$, let $t'$ be constructed as described above, and let $h\colon t' \rightarrow t$ be the function defined in the construction. 
\begin{enumerate}

\item\label{lemma_cuttingproponce_strattree}
$t'$ is a strategy tree.

\item\label{lemma_cuttingproponce_vecbounds}
 $\waitvec(w) \le \waitvec(h(w))$ for every $w \in t'$.

\item\label{lemma_cuttingproponce_waitbounds} 
If $\sigma(t)$ bounds the waiting times for condition~$j'$ by $b$, then so does $\sigma(t')$.

\item\label{lemma_cuttingproponce_valbounds}
 $\val(\sigma(t'), v) \le \val(\sigma(t), v)$ for every $v \in W_0(\game)$.
 
\end{enumerate}
\end{lmm}

\begin{proof}
(\ref{lemma_cuttingproponce_strattree}) Prefix-closure and the first requirement on a strategy tree can be proven by a straightforward induction over the length of $w \in t'$ while the second requirement is satisfied by construction. Now, consider $w_0 \cdots w_n \in t'$ with $w_n \in V_0$. We have $\last(h(w_0 \cdots w_n)) = w_n$, i.e., there is a unique successor $v$ of $w_n$ with $h(w_0 \cdots w_n) v \in t$. By construction, $w_0 \cdots w_n v$ is added to $t'$, but no $w_0 \cdots w_n v'$ for $v \not= v'$. Hence, there is a unique $v$ such that $w_0 \cdots w_n v \in t'$, i.e., the third requirement is satisfied. The reasoning for the fourth requirement is dual.

(\ref{lemma_cuttingproponce_vecbounds}) By induction over $|w|$. The claim is trivially true for $|w| \le 1$. Thus, consider $w = w'v \in t$. If $h(w'v) = h(w')v$ then 
\[ \waitvec(w) = \waitvec(w'v) \le \waitvec(h(w')v)  = \waitvec(h(w'v)) \enspace,\]
where the inequality follows from an application of Remark~\ref{rem_wtevolution} to the induction hypothesis $\waitvec(w') \le \waitvec(h(w'))$. On the other hand, if $h(w'v) = x$ for some $x$ satisfying $\waitvec(h(w')v) \le \waitvec(x)$ then \[ \waitvec(w) = \waitvec(w'v) \le \waitvec(h(w')v) \le  \waitvec(x) = \waitvec(h(w'v)) \enspace, \]
where the first inequality again follows from an application of Remark~\ref{rem_wtevolution} to the induction hypothesis.

(\ref{lemma_cuttingproponce_waitbounds}) This follows directly from Item~(\ref{lemma_cuttingproponce_vecbounds}).

(\ref{lemma_cuttingproponce_valbounds}) Let $\rho = \rho_0 \rho_1 \rho_2 \cdots$ be consistent with $\sigma(t')$ and consider the sequence~\[ h(\rho_0), h(\rho_0\rho_1), h(\rho_0\rho_1\rho_2), \ldots \]
of elements from $t$, which is an increasing chain in the (strict) prefix relation. Hence, the sequence has a unique limit $h(\rho) = \rho_0' \rho_1' \rho_2' \cdots \in V^\omega$ such that $h(\rho_0 \cdots \rho_n)$ is a prefix of $h(\rho)$ for every $n$, which is the play from which $\rho$ is obtained by removing loops. The limit is consistent with $\sigma(t)$ as it is a path through $t$. 

Let 
\[ R = \set{ \ell \in \nats \mid \text{there is no $n$ with }h(\rho_0 \cdots \rho_{n}) = \rho_0' \cdots \rho_\ell'  } \]
be the positions of vertices of $h(\rho)$ that are removed. Due to Remark~\ref{remark_cuttingprop}(\ref{remark_cuttingprop_largewaittime}), we have $\waittime_j(\rho_0' \cdots \rho_\ell') > f^{-1}_j(\val_\game)$ for every $\ell \in R$ and therefore 
\[ \pnlty(\rho'_0 \cdots \rho'_\ell) \ge f_j(\waittime_j(\rho'_0 \cdots \rho'_\ell)) > \val_\game \ge \val(\sigma(t), \rho_0')  \ge \val(h(\rho))\enspace.\] 
Thus, 
\begin{equation} \limsup_{n \rightarrow \infty} \frac{1}{n} \sum_{\ell = 0}^{n-1} \pnlty(h(\rho_0 \cdots \rho_\ell)) \le \val(h(\rho))
 	 \enspace,
\label{eq_boundonlimit} \end{equation}
since the average (and therefore also the limit superior of the averages) only decreases when we omit summands which are larger than the limit superior of the averages, i.e., those for $\rho_0 \cdots \rho_\ell$ with $\ell \in R$. 
Furthermore, we have $\waitvec(\rho_0 \cdots \rho_\ell) \le \waitvec(h(\rho_0 \cdots \rho_\ell))$ for every $\ell$ due to Item~(\ref{lemma_cuttingproponce_vecbounds}), and therefore 
\[
\frac{1}{n}\sum_{\ell = 0}^{n-1} \pnlty(\rho_0 \cdots \rho_{\ell})
\le
\frac{1}{n}\sum_{\ell = 0}^{n-1} \pnlty(h(\rho_0 \cdots \rho_{\ell})) \enspace.
\]
Thus, the value of $\rho$, the limit superior of the left-hand side of the inequality is smaller or equal to the limit superior of the right-hand side, which in turn is smaller or equal to the value of $h(\rho)$, as shown in Equation~(\ref{eq_boundonlimit}). Thus, we have $\val(\rho) \le \val(h(\rho))$.

Now, we can lift this upper bound to the values of the strategies: we have 
\begin{align*} \val(\sigma(t'), v) = 
\sup\limits_{\rho\in\behavior(v, \sigma(t'))} \val(\rho) \le& 
\sup\limits_{\rho\in\behavior(v, \sigma(t'))} \val(h(\rho)) \\
\le&
\sup\limits_{\rho\in\behavior(v, \sigma(t))} \val(\rho) = 
\val(\sigma(t), v) \enspace,
 \end{align*}
where the first inequality is the one just proven above and the second one due to the fact that $h(\rho)$ is consistent with $\sigma(t)$.
\end{proof}

From now on denote the tree~$t'$ by $I_j(t)$ as it is obtained by removing loops w.r.t.\ condition~$j$ from $t$. Note that we have not claimed that $I_j(t)$ bounds the waiting times for condition~$j$. We will now apply $I_j$ infinitely often and show that the limit of the trees obtained this way does indeed bound the waiting times. Whether applying $I_j$ once suffices to achieve this is an open question.

Formally, given a strategy tree~$t$ we define an infinite sequence of trees via $t_0 = t$ and $t_{n+1} = I_{j}(t_{n})$. Furthermore, for $n>0$ let $h_n \colon t_n \rightarrow t_{n-1}$ be the function constructed in the definition of $t_n$. We define the limit~$t_\omega$ of the $t_n$ as follows: $w \in t_\omega$ if and only if $h_{n}(w) = w$ for almost all $n$, i.e., for all but finitely many $n$. Note that $h_n(w) = w$ implies $w \in t_n \cap t_{n-1}$. 

\begin{lmm}\label{lemma_cuttingproplimit}
Let $t$ be a strategy tree such that $\val(\sigma(t) ,v ) \le \val_\game$ for every $v \in W_0(\game)$ and let $t_\omega$ be constructed as described above.
\begin{enumerate}

\item\label{lemma_cuttingproplimit_strattree}
$t_\omega$ is a strategy tree.

\item\label{lemma_cuttingproplimit_newbound}
$\sigma(t_\omega)$ bounds the waiting times for condition~$j$ to $f^{-1}_j(\val_\game) + \bnd(s,k-1)$.

\item\label{lemma_cuttingproplimit_waitbounds}
If $\sigma(t)$ bounds the waiting times for condition~$j'$ by $b$, then so does $\sigma(t_\omega)$.

\item\label{lemma_cuttingproplimit_valbounds}
$\val(\sigma(t_\omega), v) \le \val(\sigma(t), v)$ for every $v \in W_0(\game)$.
 
\end{enumerate}
\end{lmm}

\begin{proof}
(\ref{lemma_cuttingproplimit_strattree}) Prefix-closure of $t_\omega$ follows from Remark~\ref{remark_cuttingprop}.(\ref{remark_cuttingprop_prefixclosure}): let $w \in t_\omega$, i.e., we have $h_n(w) = w$ for almost all $n$. Then, we have $h_n(w') = w'$ for the same $n$ and all prefixes~$w'$ of $w$. Hence, $w' \in t_\omega$. 

Furthermore, $w$ satisfies the first requirement on a strategy tree, since $w$ is in some $t_n$, for which the first requirement holds due to Lemma~\ref{lemma_cuttingproponce}.(\ref{lemma_cuttingproponce_strattree}). Furthermore, the second requirement is satisfied by construction: every $t_n$ contains $W_0(\game)$ and we have $h_n(v) = v$ for every $v \in W_0(\game)$. Hence, it remains to prove the last two properties. 

For the third requirement, consider $w \in t_\omega$ with $\last(w) \in V_0$. We have to show that there is a unique $v$ with $wv \in t_\omega$. Let $n_w$ be such that $h_n(w) = w$ for every $n \ge n_w$, which implies $w \in t_n$ for every $n \ge n_w -1$. 

As every $t_n$ is a strategy tree, there is a unique $v_n$ with $wv_n \in t_n$ for every $n \ge n_w -1$. We claim $v_n = v_{n_w-1}$ for every $n \ge n_w -1$. The induction start $n = n_w-1$ is trivial, so consider some $n > n_w-1$: $wv_{n}$ is in $t_n$, since $h_{n}(w)v_n = w v_n$ is in $t_{n-1}$. Now, this implies $v_n = v_{n-1}$, as $v_{n-1}$ is the unique vertex~$v$ with $wv \in t_{n-1}$. An application of the induction hypothesis $v_{n-1} = v_{n_w-1}$ yields the desired result. 

From now on we denote $v_{n_w-1}$ by $v$. We have to show $h_n(wv) = wv$ for almost all $n$. As a first case, assume we have $\waittime_j(wv) \le f^{-1}_j ( \val_\game )$. Then, we have $h_n(wv) = wv$ for every $n \ge n_w -1$, as we are in case (1) of the definition of $h_n(wv)$. 
Now, assume we have $\waittime_j(wv) > f^{-1}_j ( \val_\game )$. We claim 
\[ | S_j(t_n, wv) | \ge | S_j(t_{n+1}, wv) |\]
for every $n \ge n_w -1$.  Every element in $S_j(t_{n+1}, wv)$ is mapped by $h_{n+1}$ to an element in $S_j(t_{n}, wv)$. Hence, finiteness of the sets and  injectivity of $h_{n+1}$
proves our claim. Furthermore, we have equality \[ | S_j(t_n, wv) | = | S_j(t_{n+1}, wv) |\]
only in case $h_{n+1}(wv) = wv $: if $h_{n+1}(wv) \not= wv$, then there is no element in $S_j(t_{n+1}, wv)$ that is mapped to $wv \in S_j(t_{n}, wv)$, due to Remark~\ref{remark_cuttingprop}.(\ref{remark_cuttingprop_subplay}).

Thus, the sequence $(S_j(t_n, wv))_{n \ge n_w-1}$ gets stationary and from that point onwards, we have $h_n(wv) = wv$. Thus, $wv \in t_\omega$. Furthermore, $v$ is unique since $wv' \in t_\omega$ with $v' \not= v$ implies that there is an $n$ with $wv \in t_n$ and $wv' \in t_n$, which contradicts the fact that $t_n$ is a strategy tree.

The fourth and final requirement on $t_\omega$ can be proven dually: let $w \in t_\omega$ with $\last(w) \in V_1$. We have to show that $wv \in t_\omega$ for every successor~$v$ of $\last(w)$. Again, every such $wv$ is in every $t_n$ for $n \ge n_w -1$. Now, using the same reasoning as for the third requirement, one can show $wv \in t_\omega$ for every $v$.

(\ref{lemma_cuttingproplimit_newbound}) Assume there is a $w = w_0 \cdots w_m \in t_\omega$ with $\waittime_j( w ) > f^{-1}_j(\val_\game) + \bnd(s,k-1)$. Let $w' = w_0 \cdots w_{m - \bnd(s,k-1) }$. Thus, we have $\waittime_j( w' ) > f^{-1}_j(\val_\game)$. 

Now, consider the infix~$ w_{{m - \bnd(s,k-1)+1 }} \cdots w_m$ of length~$\bnd(s,k-1)$. It contains a dickson pair~$(m_0, m_1)$ with $m - \bnd(s,k-1)+1 \le  m_0 < m_1 \le m$ by definition of $\bnd(s,k-1)$. Here, the parameter $k -1$ stems from the fact that the waiting times for condition~$j$ increase throughout the infix, i.e., there are only $k-1$ conditions we have to consider to obtain a dickson pair. 

Now, consider an index~$n$ with $h_n(w) = w$, which implies $h_n( w_0 \cdots w_{m_0} ) = w_0 \cdots w_{m_0} $, too. Since we have $\waittime_j( w_0 \cdots w_{m_0} ) > f^{-1}_j(\val_\game)$, we are in the second case of the definition of $h_n(w_0 \cdots w_{m_0})$ and the existence of $w_0 \cdots w_{m_1} \in S_j(t_{n-1}, w_0 \cdots w_{m_0})$ with $w_{m_0} = w_{m_1}$ and $\waitvec( w_0 \cdots w_{m_0}) \le \waitvec( w_0 \cdots w_{m_1})$ implies $h_n(w_0 \cdots w_{m_0}) \not= w_0 \cdots w_{m_0}$, i.e., we have derived the desired contradiction.

(\ref{lemma_cuttingproplimit_waitbounds}) 
For every $w \in t_\omega$ there is an $n_w \in \nats$ such that $h_n(w) = w$ for every $n \ge n_w$. Furthermore, due to Remark~\ref{remark_cuttingprop}.(\ref{remark_cuttingprop_prefixclosure}) we can pick the $n_w$ in way that they satisfy $n_{w'} \le n_w$ for every $w,w'$ such that $w'$ is a prefix of $w$.

Now, define $h_\omega \colon t_\omega \rightarrow t$ via 
\[  h_\omega(w) = h_1 ( h_2( \cdots h_{n_w-2}(h_{n_w-1} (w)) \cdots )) \enspace. \]
Applying Lemma~\ref{lemma_cuttingproponce}.(\ref{lemma_cuttingproponce_waitbounds}) 
inductively yields $\waitvec(w) \le \waitvec(h_\omega(w))$ for every $w \in t_\omega$. The result follows.

(\ref{lemma_cuttingproplimit_valbounds})
The proof is analogous to the one for Lemma~\ref{lemma_cuttingproponce}.(\ref{lemma_cuttingproponce_valbounds}), we just have to replace $h$ by $h_\omega$.
\end{proof}

We denote the limit $t_\omega$ of the applications of $I_j$ to $t$ by $I_{j, \omega}(t)$. Now, we are ready to prove the main result of this subsection.

\begin{proof}[Proof of Lemma~\ref{lemma_bounding}]
Consider the strategy tree $t' = I_{k, \omega}(\cdots I_{2, \omega}(I_{1, \omega}(t(\sigma))) \cdots ) $ and the resulting strategy~$\sigma' = \sigma(t')$. An inductive application of Lemma~\ref{lemma_cuttingproplimit} yields that $\sigma'$ bounds the waiting times for every condition~$j$ by $f^{-1}_j(\val_\game) + \bnd(s,k-1) $ and satisfies $\val(\sigma', v) \le \val(\sigma, v)$ for every $v \in V$.
\end{proof}

The construction presented here gives a correction to the one presented in~\cite{HornThomasWallmeier08} where each loop removal operator~$I_j$ is applied only once. 
\section{Computing Optimal Strategies for RR Games}
\label{sec_computingoptimalstrategies}

In this section, we prove our main result: Player~$0$ has optimal finite-state
winning strategies in RR games, which are effectively computable. To this end,
we construct a mean-payoff game in an arena which keeps track of the waiting
times up to the bounds~$f^{-1}_j(\val_\game) + \bnd(s, k-1)$ and whose weight function reflects the penalty functions.

Then, we prove that an optimal strategy for the mean-payoff game, which always exists, induces an
optimal winning strategy for the RR game. This approach is complete due to the
fact that in an RR game an optimal strategy can be assumed to
have bounded waiting times. We begin by introducing mean-payoff games in
Subsection~\ref{subsection_mpgdef} and then prove our main result in
Subsection~\ref{subsection_mpgreduction}.

\subsection{Mean-Payoff Games}
\label{subsection_mpgdef}

A mean-payoff game~$\game = (\arena, \weight)$ consists of an arena~$\arena$
with set~$E$ of edges and a weight function~$\weight\colon E \rightarrow
\set{-W, \ldots , W}$ for some $W \in \nats$. Given a play $\rho = \rho_0
\rho_1 \rho_2 \cdots$ we define its value for Player~$0$ as 
\[\mpgval_0(\rho) = \limsup_{n\rightarrow \infty} \frac{1}{n}\sum_{\ell =1}^n
\weight(\rho_{\ell-1}, \rho_\ell)\enspace,\]
and its value for Player~$1$ as
\[\mpgval_1(\rho) = \liminf_{n\rightarrow \infty} \frac{1}{n}\sum_{\ell =1}^n
\weight(\rho_{\ell-1}, \rho_\ell)\enspace.\]
Intuitively, Player~$0$ wants to minimize $\mpgval_0(\rho)$ while Player~$1$ wants to maximize $\mpgval_1(\rho)$. Note that we always have $- W \le
\mpgval_1(\rho) \le \mpgval_0(\rho) \le W$. For notational convenience we have
swapped the roles of the players, i.e., classically Player~$0$'s value is the
$\liminf$ and Player~$1$'s value is the $\limsup$ of the mean weights.

\begin{thrm}[\hspace{-.00001cm}\cite{EhrenfeuchtM79, ZwickP96}]
\label{theorem_mpg}
For every mean-payoff game there exist positional strategies~$\sigmaopt$ for
Player~$0$ and $\tauopt$ for Player~$1$ and values~$\nu(v)$ for every 
vertex~$v$ such that
\begin{enumerate}

\item every play~$\rho \in \behavior(v, \sigmaopt)$ satisfies
$\mpgval_0(\rho) \le \nu(v)$, and
	
\item every play~$\rho \in \behavior(v, \tauopt)$ satisfies $\mpgval_1(\rho)
\ge \nu(v)$.

\end{enumerate}
The strategies and values are computable in pseudo-polynomial time (i.e., in polynomial time in the
size of the arena and in the maximal weight of an edge).
\end{thrm}

Especially, the unique play $\rho \in \behavior(v, \sigmaopt) \cap
\behavior(v, \tauopt)$ satisfies $\mpgval_0(\rho) = \mpgval_1(\rho) = \nu(v)$.
The strategy~$\sigmaopt$ is optimal in the sense that there is no strategy for
Player~$0$ that guarantees a strictly smaller value than $\nu(v)$ when
starting from $v$. The analogous statement is true for $\tauopt$.
\subsection{Computing Optimal Strategies for RR Games via Mean-Payoff Games}
\label{subsection_mpgreduction}

In this subsection, we prove our main theorem: optimal strategies for RR games exist and can be effectively computed via the solution of a single mean-payoff game.

\begin{thrm}
\label{theorem_main}
In every RR game, Player~$0$ has an optimal finite-state winning strategy,
which is effectively computable.
\end{thrm}

\begin{proof}
Let $\game = (\arena, \rrcond)$ be an RR game with $s$ vertices and $k$ RR conditions and let $f_j$ be a strictly increasing penalty function for every $j$. Define \[\tmaxj = f^{-1}_j(\val_\game) + \bnd(s, k-1)\enspace,\] which satisfies $\tmaxj \ge 1$.
Now, let $\mem = (M, \init, \update)$ where \[M = \left(\prod_{j \in \natspref{k}}\set{ 0, \ldots, \tmaxj }\right) \cup
\set{\bot}\] is the set of all waiting time vectors whose values are bounded by
$\tmaxj$ in coordinate~$j$ with an additional element~$\bot$ denoting that the bound $\tmaxj$ is
exceeded for some $j$. Furthermore, we define
$\init(v) = (t_1, \ldots, t_k)$ with
\[ t_j = \begin{cases}
1 &\text{if $v \in Q_j \setminus P_j$,}\\
0 &\text{otherwise,}
\end{cases}  \]
and $\update(\bot, v) = \bot$. It remains to define $\update((t_1, \ldots,
t_k), v)$: if there is a $j$ such that $t_j =\tmaxj$ and $v \notin P_j$, then we
define $\update((t_1, \ldots, t_k), v) = \bot$. Otherwise, we have
$\update((t_1, \ldots, t_k), v) = (t_1', \ldots, t_k')$ with (cf. the definition of the waiting time~$\waittime$)
\[ t_j' = \begin{cases}
0 		&\text{if $t_j = 0$ and $v \notin Q_j \setminus P_j$,}\\
1 		&\text{if $t_j = 0$ and $v \in Q_j \setminus P_j$,}\\
0 		&\text{if $t_j > 0$ and $v \in P_j$,}\\
t_j+1	&\text{if $t_j > 0$ and $v \notin P_j$.}
\end{cases}  \]
Each $t_j'$ is again bounded by $\tmaxj$. Intuitively, the memory
keeps track of the waiting times of play prefixes up to the thresholds~$\tmaxj$. If
a threshold is exceeded, a sink state is reached.

We define the mean-payoff game~$\game' =
(\arena \times \mem, \weight)$ by 
\[ \weight (
(v, (t_1, \ldots, t_k)) , (v', m)
 ) = \sum_{j \in \natspref{k}} f_j( t_j ) \]
for every memory state~$m \in M$ and $\weight((v,\bot), (v', \bot)) = 1 +
\sum_{j \in \natspref{k}} f_j(\tmaxj)$. Due to $f_j$ being strictly increasing, the maximal edge weight in $\game'$ is $1 + \sum_{j \in \natspref{k}} f_j(\tmaxj)$, which appears only on the edges between vertices of the form~$(v, \bot)$. We continue by stating some simple connections between plays in $\game$ and their
extended plays in $\game'$.

\begin{rmrk}
\label{rem_waittimememory}
Let $\rho = \rho_0 \rho_1 \rho_2 \cdots$ be a play
in $\game$ and
$\extend(\rho)$ its extended play in~$\game'$.
\begin{enumerate}

\item If $\update^*(\rho_0 \cdots \rho_n) \not= \bot$, then $\update^*(\rho_0
\cdots \rho_n) = \waitvec(\rho_0 \cdots \rho_n)$.

\item If $\update^*(\rho_0 \cdots \rho_n) = \bot$, then there is a
prefix~$\rho_0 \cdots \rho_{p}$ of $\rho_0 \cdots \rho_n$ and an
index~$j$ such that $\waittime_j(\rho_0 \cdots \rho_{p}) > \tmaxj$ and every
suffix~$\rho_0 \cdots \rho_{p} \cdots \rho_{s}$ of $\rho_0 \cdots \rho_{p}$
satisfies $\update^*(\rho_0 \cdots \rho_{p} \cdots \rho_{s}) = \bot$.

\item\label{rem_waittimememory_boundedwaittimeval}
 If $\extend(\rho)$ does not visit the memory state~$\bot$, then 
$\val(\rho) = \mpgval_0(\extend(\rho)) < 1 + \sum_{j \in \natspref{k}} f_j(\tmaxj)$, i.e.,
the value of the play~$\rho$ in the RR game and the value of its extended
play~$\extend(\rho)$ in the mean-payoff game are equal (and smaller than the weight of the edges between the sink states with memory~$\bot$)
if the waiting times are bounded by $\tmaxj$.

\item If $\extend(\rho)$ visits the memory state~$\bot$, then
$\mpgval_0(\extend(\rho)) = 1 + \sum_{j \in \natspref{k}} f_j(\tmaxj)$.
\end{enumerate} 
\end{rmrk}

Now, we can begin with the actual proof of Theorem~\ref{theorem_main}, in which we have to deal with several
strategies for Player~$0$. Throughout the proof, 
we denote strategies for $\game$ without a prime and strategies for $\game'$ with a prime. 
The strategies always come in pairs, one for the RR game~$\game$ and one for the mean-payoff game~$\game'$. 

\begin{description}
\item[\boldmath$\sigma$ and $\sigma'$] $\sigma$ uniformly bounds the waiting times in $\game$. This strategy is turned into $\sigma'$ for $\game'$ which never reaches the memory state~$\bot$. This bounds the values~$\mpgval(v)$ of the game~$\game'$.

\item[\boldmath$\sigmaopt$ and $\sigmaoptprime$] $\sigmaoptprime$ is an optimal strategy for $\game'$, which is turned
 into a strategy~$\sigmaopt$ for $\game$. Due to the properties of $\sigma'$, we know that $\sigmaoptprime$ never reaches the memory state~$\bot$, which in turn bounds the waiting times of $\sigmaopt$. Then, we show that $\sigmaopt$ is indeed optimal.

\item[\boldmath$\sigmaopthat$ and $\sigmaopthatprime$] To this end, we assume it is not optimal, i.e., there is
a better strategy~$\sigmaopthat$. This is turned into a strategy~$\sigmaopthatprime$
for $\game'$, which is strictly better than the optimal strategy~$\sigmaoptprime$. This 
contradiction finishes the proof.

\end{description}

Due to Corollary~\ref{corollary_upperboundonval} and
Lemma~\ref{lemma_bounding}, there is a strategy~$\sigma$ for Player~$0$ for
$\game$ such that 
$\val(\sigma, v) \le \sum_{j \in \natspref{k}}f_j(s k 2^k)$ for every $v \in W_0(\game)$ and such that 
$\waittime_j(w) \le \tmaxj$ for every play prefix~$w$ that is consistent with
$\sigma$ and starts in $W_0(\game)$. First, we turn $\sigma$ into a
strategy~$\sigma'$ for $\game'$ and use Remark~\ref{rem_waittimememory} to relate their values.
To this end, let
\begin{equation}
\sigma'((v_0, m_0) \cdots (v_n, m_n)) = (\sigma(v_0 \cdots v_n), \update(m_n,
\sigma(v_0 \cdots v_n) ) \enspace,
\label{eq_sigmatosigma'}
\end{equation}
i.e., we mimic the behavior of $\sigma$ in the first component and update the memory state in the second component
accordingly. Let $\rho' = (v_0, m_0) (v_1, m_1) (v_2, m_2) \cdots$ be
consistent with $\sigma'$. A straightforward induction shows that
$\project(\rho') = v_0 v_1 v_2 \cdots$ is consistent with $\sigma$.
Also, if $(v_0, m_0) = (v_0, \init(v_0))$
then $\rho' = \extend(\project(\rho'))$. If additionally $v_0 \in W_0(\game)$ then an application of
Remark~\ref{rem_waittimememory}(\ref{rem_waittimememory_boundedwaittimeval})
yields \[\val(\project(\rho')) = \mpgval_0(\rho') < 1 + \sum_{j \in
\natspref{k}} f_j(\tmaxj)\enspace.\]
 The following is now immediate.

\begin{rmrk}
	\label{rem_nosink}
If  $v \in W_0(\game)$, then $\mpgval(v, \init(v)) < 1 + \sum_{j \in
\natspref{k}} f_j(\tmaxj)$. 
\end{rmrk}

Now, consider an optimal strategy~$\sigmaoptprime$ for Player~$0$ in $\game'$ as
guaranteed by Theorem~\ref{theorem_mpg}. Due Remark~\ref{rem_nosink}, every play
that starts in a vertex of the form~$(v, \init(v))$ for some $v \in W_0(\game)$
never visits the memory state~$\bot$. Now, let $\sigmaopt$ be the strategy for
$\game$ induced by $\sigmaoptprime$ with memory~$\mem$. Formally, we
define it by giving a next-move function via 
$\nxt(v, m) = v'$ in case we have $\sigmaoptprime(v,m) = (v', m')$ for some $m'$.
Let $\rho$ be a play in $\game$ that is consistent with $\sigmaopt$ and starts
in $W_0(\game)$. A straightforward induction shows that $\extend(\rho)$ (which
starts in $(v, \init(v))$) is consistent with $\sigmaoptprime$. Thus, the
memory state~$\bot$ is never reached and we have $\val(\rho) =
\mpgval_0(\extend(\rho))$.

We claim that $\sigmaopt$ has the desired properties: it is
finite-state and effectively computable. Hence, it remains to show that it is
optimal. Assume it is not. Then, there exists a vertex~$v$ and a
strategy~$\sigmaopthat$
for Player~$0$ in $\game$ such that $\val(\sigmaopthat, v) < \val(\sigmaopt,
v) \le \val_\game$. Due to Lemma~\ref{lemma_bounding}, we can assume that $\sigmaopthat$ bounds
the waiting times for every condition~$j$ by $\tmaxj$. Now, using the same definition as in
(\ref{eq_sigmatosigma'}), we turn $\sigmaopthat$ into a
strategy~$\sigmaopthatprime$ for Player~$0$ in $\game'$.

As above, for every play~$\rho' = (v_0, m_0) (v_1, m_1) (v_2, m_2) \cdots$
that is consistent with $\sigmaopthatprime$ the projected play~$v_0 v_1
v_2 \cdots$ is consistent with $\sigmaopthat$.
Furthermore, if $(v_0, m_0) = (v_0, \init(v_0))$, then $\rho' =
\extend(\project(\rho'))$ and $\mpgval_0(\rho') = \val(\project(\rho'))$.

Recall that $v$ is the vertex of $\game$ from which $\sigmaopthat$ is better
than $\sigmaopt$. Now, consider the optimal strategy~$\tauoptprime$ for
Player~$1$ in $\game'$ (as in Theorem~\ref{theorem_mpg}) and let $\rho'$ be the unique play in $\game'$ that
starts in $(v, \init(v))$ and is consistent with both $\sigmaopthatprime$ and
$\tauoptprime$. We have
\[
\mpgval_0(\rho') = \val(\project(\rho')) \le \val(\sigmaopthat, v) <
\val(\sigmaopt, v) \le \mpgval(v, \init(v)) \le \mpgval_1(\rho')  \enspace,
\]
which yields the desired contradiction to the fact that we have
$\mpgval_0(\rho') \ge \mpgval_1(\rho')$ by definition. Here, the 
inequality~$\val(\sigmaopt, v) \le \mpgval(v, \init(v))$ follows from the fact that every
play that contributes to $\val(\sigmaopt, v)$ has an extended play in $\game'$
that starts in $(v, \init(v))$, is consistent with $\sigmaopt'$, and has the
same value (for Player~$0)$, which is smaller than $\mpgval(v, \init(v))$ by
Theorem~\ref{theorem_mpg}. 
\end{proof}

\section{Conclusion}
\label{sec_conc}

We have presented an algorithm that computes optimal winning strategies for RR games in case the quality of a play is measured by the limit superior of the mean accumulated penalties on the waiting times between requests and their responses. To this end, we proved that the waiting times of strategies with small value can be assumed to be bounded by some doubly-exponential bound. Thus, the search space for an optimal winning strategy is finite and the problem of finding one can be reduced to computing an optimal strategy for a mean-payoff game. 

The reduction presented here is also applicable to a more general winning condition, the so-called poset condition~\cite{Zimmermann09}, where a request has to answered by a partially ordered set of events. In such games, the waiting times are used to measure the time between a request and the occurrence of the last event required in its response. Here, unlike in RR games, we measure the waiting time for every request, even if there is currently an open one. This is necessary since a new request might appear while an old request is already partially answered, i.e., satisfying the remaining events that answer the old request does not suffice to answer the new one. This situation cannot occur in RR games, as a response is a single event. Thus, by viewing an RR game as a poset game, we can also compute an optimal strategy when measuring the quality by taking all requests into account.

Unfortunately, the reduction to mean-payoff games presented here is expensive in terms of running time of the algorithm and also in terms of the memory requirements of the optimal strategy: the size of the mean-payoff game is doubly-exponential and the largest weight in this game is doubly exponential (if the penalty functions are the identity function, otherwise these values are even larger). The best algorithms for mean-payoff games have a polynomial running time in these two parameters. Thus, our algorithm has a doubly-exponential running time. This has to be contrasted with the $\exptime$-completeness of computing an arbitrary winning strategy for an RR game~\cite{ChatterjeeHenzingerHorn11}. 
Furthermore, the size of the memory structure implementing the optimal strategy for the RR game computed by our algorithm is also at least of doubly-exponential size, again larger than arbitrary winning strategies, which are of exponential size~\cite{ChatterjeeHenzingerHorn11, WallmeierHuttenThomas03}.

As mentioned earlier, the upper bound on the waiting times is tight as shown by~\cite{CzerwinskiGK14}. Hence, to obtain a faster algorithm and smaller optimal winning strategies, a different approach is necessary. The exact complexity of computing optimal strategies is an open problem. Another approach to overcome the high complexity is to consider heuristics and approximation algorithms, which compute strategies that realize the value of an optimal strategy up to a certain factor. Finally, the size of the optimal strategy computed here is much larger than the lower bounds on memory requirements in RR games. This raises the question whether there is a tradeoff between the size and the quality of a strategy.

\bibliographystyle{plain}
\bibliography{biblio}
\end{document}